\documentclass[a4paper, 11pt]{article}
\usepackage[utf8]{inputenc}
\usepackage[top=1in, bottom=1.2in, left=1in, right=1in]{geometry}
\usepackage{amsmath, amssymb, amsthm}
\usepackage{enumerate}
\usepackage{hyperref}
\usepackage{array}

\newcommand{\T}{{\rm T}}
\newcommand{\cc}{\mathbb{C}}
\newcommand{\g}{\mathcal{G}}
\newcommand{\h}{\mathcal{H}}
\newcommand{\s}{\mathcal{S}}
\newcommand{\qc}{\mathcal{C}}
\newcommand{\nn}{\mathcal{N}}
\newcommand{\rn}{\mathcal{R}}
\newcommand{\mxn}{\mod X^n+1}
\newcommand{\ba}{\boldsymbol{a}}
\newcommand{\bb}{\boldsymbol{b}}
\newcommand{\bc}{\boldsymbol{c}}
\newcommand{\bd}{\boldsymbol{d}}
\newcommand{\bu}{\boldsymbol{u}}
\newcommand{\bv}{\boldsymbol{v}}
\newcommand{\bx}{\boldsymbol{x}}
\newcommand{\by}{\boldsymbol{y}}
\newcommand{\ze}{\boldsymbol{0}}
\newcommand{\stil}{\widetilde{S}}
\newcommand{\ifif}{if and only if }
\newcommand{\wrt}{with respect to }
\newcommand{\fp}{\mathbb{F}_p}
\newcommand{\fpn}{\mathbb{F}_p^n}
\newcommand{\fps}{\mathbb{F}_{p^2}}
\newcommand{\fpe}{\mathbb{F}_p(\eta)}
\newcommand{\fpsn}{\mathbb{F}_{p^2}^n}
\newcommand{\fpen}{\mathbb{F}_p(\eta)^n}
\newcommand{\htn}{\mathcal{H}^{\otimes^n}}
\newcommand{\fpns}{\mathbb{F}_p^n\times\mathbb{F}_p^n}
\newcommand{\ssm}{\mbox{{\footnotesize\textbackslash}}}

\newcommand{\ket}[1]{\left| #1 \right\rangle}
\newcommand{\braket}[1]{\left\langle #1 \right\rangle}
\newcommand{\ideal}[1]{\left\langle #1 \right\rangle}

\newtheorem{definition}{Definition}[section]
\newtheorem{proposition}{Proposition}[section]
\newtheorem{corollary}[proposition]{Corollary}
\newtheorem{theorem}[proposition]{Theorem}
\newtheorem{claim}[proposition]{Claim}
\newtheorem{lemma}[proposition]{Lemma}

\numberwithin{equation}{section}

\title{{\bf $t$-Frobenius Negacyclic Codes}}
\author{Priyabrata Bag \& Santanu Dey}

\date{}

\begin{document}
\maketitle

\begin{abstract}
\noindent Let $\fp$ denote the finite field of order $p$, where $p$ is an odd prime.
We study certain quantum negacyclic codes over $\fp$ which we call $t$-Frobenius
negacyclic codes. We obtain a criterion for constructing such codes from certain
subspaces of $\fpn\times\fpn$ in terms of generating pairs of ideals of cyclotomic
rings and we completely classify all linear $t$-Frobenius negacyclic codes. Further,
the notion of BCH distance is extended for these codes and several new codes are
identified.\\
\mbox{}
\mbox{}
\end{abstract}

\section{Introduction}
\label{intro}
The class of quantum stabilizer codes is an important and well studied class of quantum
codes (\cite{s-srd}, \cite{cs-gqc}, \cite{g-qc}). Calderbank {\it et al.} in \cite{crss-gf4}
developed an elegant theory of constructing binary quantum stabilizer codes based on
classical codes and this class of quantum codes subsumed many previously known quantum codes
like those constructed in \cite{s-srd}, \cite{cs-gqc}, \cite{g-qc}, etc.
A $q$-ary version of these codes were first developed in \cite{ak-nqsc}.

To any quantum stabilizer code of length $n$, i.e., codes on the $p^n$-dimensional Hilbert
space $\h=\cc^p\otimes\cdots\otimes\cc^p$, there is an associated subspace of $\fpn\times\fpn$
and this subspace can be treated like a classical code. There is a symplectic inner product
on $\fpns$ such that for every quantum stabilizer code the associated subspace of $\fpns$
is totally isotropic (see, Definition~\ref{tis}). We introduce negacyclic quantum stabilizer codes
on the Hilbert space $\h=\cc^p\otimes\cdots\otimes\cc^p,\,p>2$ and show that for the negacyclic
stabilizer codes these subspaces of $\fpn\times\fpn$ turns out to be simultaneously negacyclic
(see, Definition~\ref{sns}). Recall that classical negacyclic codes are ideal of the
cyclotomic ring $\fp[X]/\ideal{X^n+1}$ which is a principal ideal ring. If the stabilizer
negacyclic code is linear, then the associated subspace of $\fpn\times\fpn$ is an ideal of
$\fps[X]/\ideal{X^n+1}$ and so is its dual with respect to the symplectic inner product.
We achieve a significant simplification
by restricting ourself to those negacyclic quantum stabilizer codes where the length $n$ divides
$p^t+1$ for some positive integer $t$, such that $\frac{p^t+1}{n}$ is an odd integer. We call
such codes as {\it $t$-Frobenius negacyclic codes} (cf. Definition~\ref{tfnc}).
The Frobenius automorphism over $\fp$ is used
to achieve several simplifications in characterizing such codes.

In Section~\ref{LtFNC}, we treat the case of linear stabilizer negacyclic codes and prove a certain necessary
condition for these codes in terms of the generator polynomials for the ideals, which are identified
with subspaces of $\fpn\times\fpn$ associated to the codes, of $\fps[X]/\ideal{X^n+1}$. We obtain a
factorization of any such generator polynomial. That this necessary condition is also sufficient for
a subspace of $\fpn\times\fpn$ to be the associated subspace of a quantum linear negacyclic code is
shown in Section~\ref{CLNlC}. We carry out investigation on a similar line for nonlinear negacyclic codes and
extended the notions of BCH distance to this case. With the help of a result from \cite{ks-pcc} we define
the BCH distance for any such stabilizer negacyclic code. The BCH distance is a lower bound for the
minimum distance of the stabilizer code. It is also established here that these codes cannot be
obtained using CSS construction. We list several new linear codes which are obtained by our construction.
Some new examples of nonlinear stabilizer codes are also listed.  

\section{Preliminaries}
\label{preli}
Let $p$ be a prime. Let $\h$ stand for the $p$
dimensional Hilbert space $\mathcal{L}^2(\fp)$, which is the space of all complex valued functions on $\fp$.
Then the set $\{\ket{a}:a\in\fp\}$, where $\ket{a}$ denote the characteristic function of the
singleton set $\{a\}$, forms an orthonormal basis of $\h$.
For $\ba=(a_0,a_1,\ldots,a_{n-1})^\T\in\fp^n$, let $\ket{\ba}$ denote the vector
$\ket{a_0}\otimes\ket{a_1}\otimes\cdots\otimes\ket{a_{n-1}}$ in $\htn$. The set $\{\ket{\ba}\,:
\,\ba\in\fp^n\}$ forms an orthonormal basis of $\htn$.

Let $\zeta$ be a primitive $p$-th root of unity in $\cc$. Let $\ba,\bb\in\fpn$. Define the unitary
operators $U_{\ba}$ and $V_{\bb}$ on $\htn$ by $U_{\ba}\ket{\bx}=\ket{\bx+\ba}$ and $V_{\bb}\ket{\bx}=
\zeta^{\bb^{\rm T}\bx}\ket{\bx}$ respectively. The operators of the form $U_{\ba}V_{\bb}$ are known as Weyl
operators. Two Weyl operators $U_{\ba}V_{\bb}$ and $U_{\bc}V_{\bd}$ commute \ifif
$\ba^{\rm T}\bd-\bb^{\rm T}\bc=0$.
For any two elements $\bu=(\ba,\bb)$ and $\bv=(\bc,\bd)$ in $\fpns$, the symplectic
inner product is defined by $\braket{\bu,\bv}_s=\ba^{\rm T}\bd-\bb^{\rm T}\bc$.
\begin{definition}\label{tis}
 A subset $S$ of $\fpns$ is called totally isotropic if for any two elements $\bu,\bv\in S$, the
 symplectic inner product $\braket{\bu,\bv}_s=0$.
\end{definition}
Thus for a subset $S$ of $\fpns$, the family $\{U_{\ba}V_{\bb}:(\ba,\bb)\in S\}$ of Weyl operators
is commutative \ifif the indexing set $S$ is totally isotropic. 
Let $\mathcal{W}_{n,p}$ denote the group generated by the Weyl operators on $\htn$. Then
the error group $\mathcal{E}_{n,p}$ is same as $\mathcal{W}_{n,p}$ if $p$ is odd, and is the group
generated by $\mathcal{W}_{n,p}\cup \iota\mathcal{W}_{n,p}$, $\iota=\sqrt{-1}$ in $\cc$,
when $p$ is $2$. A {\it stabilizer} code is a
subspace of $\htn$ which is invariant under some subgroup of the error group.
For a subgroup $\s$ of the error group $\mathcal{E}_{n,p}$, the stabilizer code is denoted by $\qc(\s)$
and defined by
\begin{equation}
 \qc(\s)=\{\ket{\psi}\in\htn:U\ket{\psi}=\ket{\psi}\;\forall\;U\in\s\}.
\end{equation}
It was shown by Calderbank {\it et al.} \cite{crss-og}, \cite{crss-gf4} and Gottesman \cite{g-qc} for the case when
$p=2$, and by Ashikhmin {\it et al.} \cite{ak-nqsc} and Arvind {\it et al.}
\cite{ap-wcr} for the $q$-ary case that $\qc(\s)$ is nontrivial
\ifif $\s$ is a commutative subgroup and does not contain $zI$ for any nontrivial $p$-th root of unity $z$.
We recall below a result from \cite{crss-og}, \cite{crss-gf4}, \cite{ap-wcr}
(see also \cite{dk-pt+1}) where it is shown
that such subgroups are characterized by totally isotropic subspaces of $\fpns$:

\begin{theorem}\label{cpt}
 Let $p$ be any prime, and $n$ be a positive integer. If $S$ is a totally isotropic subspace of
 $\fpns$, then the following holds:
 \begin{enumerate}[{\rm (1)}]
  \item the subset $\s=\{\omega^{\ba^{\T}\bb}U_{\ba}V_{\bb}:(\ba,\bb)\in S\}$ of unitary operators
  forms an abelian subgroup, where $\omega$ is $e^{\frac{2\pi\iota}{p}}$ when $p$ is odd and is
  $\iota=\sqrt{-1}$ when $p$ is even. Hence the invariant subspace $\qc(\s)$ of $\s$ forms a quantum
  stabilizer code.
  \item the projection operator onto the stabilizer code $\qc(\s)$ is given by $P=\sum_{U\in\s}U=$\\
  $\sum_{(\ba,\bb)\in S}\omega^{\ba^{\T}\bb}U_{\ba}V_{\bb}$.
 \end{enumerate}
\end{theorem}
Let $S$ be a subspace of $\fpns$. The dual of $S$ \wrt the symplectic inner product
is denoted by $S^{\perp}$, i.e., 
\begin{equation*}
 S^{\perp}=\{\bu=(\ba,\bb)\in\fpns: \braket{\bu,\bv}_s=0\;\forall\;\bv=(\bc,\bd)\in S\}.
\end{equation*}
It can easily be seen that $S$ is totally isotropic \ifif $S\subseteq S^{\perp}$. For a totally
isotropic subspace $S$ of $\fpns$, the dimension of $S$ is at most $n$. It is immediate that
if the dimension of $S$ is $n-k$ for some integer $k\geqslant 0$, then the dimension of $S^{\perp}$ is
$n+k$.
\begin{definition}
 For an element $(\ba,\bb)\in\fpns$, the joint weight of $\bu$ is defined by
 \begin{equation*}
  {\rm wt}(\ba,\bb)=\#\{j:(a_j,b_j)\neq(0,0),0\leqslant j\leqslant n-1\},
 \end{equation*}
 where $\ba=(a_0,a_1,\ldots,a_{n-1})$ and $\bb=(b_0,b_1,\ldots,b_{n-1})$.
\end{definition}

\begin{theorem}
 Let $S$ be a totally isotropic subspace of $\fpns$ of dimension $n-k$ for some $k\geqslant 0$.
 Then the dimension of the associated stabilizer code $\qc(\s)$ is $p^k$. Furthermore, if the joint weight of
 any element in $S^{\perp}\ssm S$ is at least $d$, then $\qc(\s)$ can correct 
 $\left\lfloor\frac{d-1}{2}\right\rfloor$ errors.
\end{theorem}
The stabilizer code $\qc(\s)$ of length $n$ and of dimension $p^k$ which
has minimum distance $d$ is denoted by $[[n,k,d]]_p$.

\section{Quantum Negacyclic Codes}
\label{QNCod}
Let $p$ be an odd prime and $n$ be a positive integer, such that
$\gcd(n,p)=1$. Let $N:\fpn\longrightarrow\fpn$ be the linear map
defined by $\bu=(u_0,u_1,\dots,u_{n-1})\longmapsto (-u_{n-1},u_0,\ldots,u_{n-2})$. Define
the operator $\nn$ on $\htn$ by $\nn\ket{\ba}=\ket{N\ba}$. Clearly $\nn$ is a unitary
operator. Now onwards $p$ will always denote an odd prime in this article.
\begin{definition}
 A quantum code $\qc$ is said to be negacyclic if it is invariant under
 the operator $\nn$.
\end{definition}

\begin{proposition}\label{pcn}
 A quantum code $\qc$ is negacyclic if and only if the projection operator onto $\qc$ commutes with $\nn$.
\end{proposition}
\begin{proof}
 Let $\g$ denote the Hilbert space $\htn$ and let $P$ be the projection operator onto $\qc$. If $P$
 commutes with $\nn$, then $\nn\qc=\nn P\g=P\nn\g$. Since $\nn$ is a unitary, $\nn\g=\g$. Thus,
 $\nn\qc=P\g=\qc$.
 
 Conversely, assume that $\nn\qc=\qc$. Then $\nn\qc^{\perp}=\qc^{\perp}$, because unitary preserves inner product.
 Let $\ket{\psi}\in\g$. Clearly $\ket{\psi}$ can be written uniquely as $\ket{\psi}=\ket{\psi_1}+\ket{\psi_2}$,
 where $\ket{\psi_1}\in\qc$ and $\ket{\psi_2}\in\qc^{\perp}$. Therefore, $\nn P\ket{\psi}=\nn\ket{\psi_1}=P
 (\nn\ket{\psi_1}+\nn\ket{\psi_2})=P\nn\ket{\psi}$. Hence, $\nn P=P\nn$. This completes the proof.
  
\end{proof}

\begin{definition}\label{sns}
 A subspace $S$ of $\fp^n\times\fp^n=\fp^{2n}$ is said to be simultaneously negacyclic if for
 any $(\ba,\bb)$ in $S$, $(N\ba,N\bb)$ is also in $S$.
\end{definition}

\begin{proposition}\label{ssp}
 If $S$ is a simultaneously negacyclic subspace of $\fpns$, then the dual $S^{\perp}$ of $S$ \wrt
 the symplectic inner product is also simultaneously negacyclic subspace of $\fpns$.
\end{proposition}
\begin{proof}
 Note that, $S^{\perp}$ is a subspace of $\fpns$.
 It is easy to check that, $N^{2n}=I_n$ and $N^{\T}=N^{-1}$. Therefore, $N^{\T}=N^{2n-1}$. Let
 $(\ba,\bb)\in S^{\perp}$ and let $(\bx,\by)\in S$. Then
 $\braket{(N\ba,N\bb),(\bx,\by)}_s=\braket{(\ba,\bb),(N^{\T}\bx,N^{\T}\by)}_s=0$, since
 $(N^{\T}\bx,N^{\T}\by)=(N^{2n-1}\bx,N^{2n-1}\by)\in S$. Hence $(N\ba,N\bb)\in S^{\perp}$.
 This completes the proof.
  
\end{proof}

\begin{proposition}
 Let $S$ be a totally isotropic subspace of $\fpns$. The stabilizer
 code $\qc(\s)$ is negacyclic \ifif $S$ is simultaneously negacyclic.
\end{proposition}
\begin{proof}
 Let $P$ be the projection operator onto $\qc(\s)$. As observed in Proposition~\ref{pcn}
 $\qc(\s)$ is negacyclic \ifif $\nn P\nn^{\dagger}=P$.
 From Theorem~\ref{cpt} we have, $P=\sum_{(\ba,\bb)\in S}\omega^{\ba^{\T}\bb}U_{\ba}V_{\bb}$.
 
 Assume that $\qc(\s)$ is negacyclic. It follows that $\nn U_{\ba}V_{\bb}\nn^{\dagger}
 =U_{N\ba}V_{N\bb}$. Thus $P=\nn P\nn^{\dagger}$ gives,
 \begin{equation*}
  P=\sum_{(\ba,\bb)\in S}\omega^{\ba^{\T}\bb}U_{N\ba}V_{N\bb}.
 \end{equation*}
 Therefore, it is necessary that for $(\ba,\bb)\in S$, $(N\ba,N\bb)$ must belongs to $S$ and hence
 $S$ in simultaneously negacyclic.
 
 Conversely, assume that $S$ is simultaneously negacyclic. Then for any $(\ba,\bb)\in S$, $(N\ba,N\bb)$ is
 also in $S$. It can easily be checked that $N^{-1}=N^{\T}$. So, $\ba^{\T}N^{\T}N\bb=\ba^{\T}\bb$.
 Combining this with the expression of $P$, we obtain $P=\nn P\nn^{\dagger}$. Hence $\qc(\s)$ is negacyclic.
  
\end{proof}

In view of this proposition, one needs to study simultaneously negacyclic subspaces of $\fpns$ to
investigate the theory of quantum stabilizer negacyclic codes. To characterize such subspaces, we first fix some
notations: Let $\rn$ denote the ring $\fp[X]/\ideal{X^n+1}$.
We identify a vector $\ba=(a_0,a_1,\dots,a_{n-1})\in\fpn$ with the polynomial
$a(X)=a_0+a_1X+\cdots+a_{n-1}X^{n-1}\in\rn$. The ring $\rn$ is isomorphic to $\fpn$ as vector
spaces over $\fp$ under this identification. For any vector $\ba\in\fpn$ we use the plain
face letter $a$ to denote the corresponding element in $\rn$. For any subset $S$ of $\fpns$ the image
in $\rn\times\rn$ will also be denoted by $S$. Note that, classical negacyclic codes
correspond to ideals of $\rn$. Since $\rn$ is a principal ideal ring,
any ideal of $\rn$ is generated by the monic polynomial of lowest degree.
It is immediate, that the generator divides $X^n+1$.

Here we obtain a criterion in terms of the polynomials for a given simultaneously negacyclic subspace $S$
of $\fpns$ to be totally isotropic. The multiplicative inverse of $X$ in $\rn$ is
$-X^{n-1}$ and is denoted by $X^{-1}$.
\begin{proposition}
 Let $S$ be a simultaneously negacyclic subspace of $\fpns$. Then $S$ is totally isotropic \ifif for any
 two elements $\bu=(\ba,\bb)$ and $\bv=(\bc,\bd)$ in $S$, the corresponding polynomials satisfy
 \begin{equation*}
  a(X)d(X^{-1})-b(X)c(X^{-1})=0 \mxn.
 \end{equation*}
\end{proposition}
\begin{proof}
 The coefficient of $X^k$ in $a(X)d(X^{-1}) \mxn$ is $\ba^{\T}N^k\bd$, and in
 $b(X)c(X^{-1}) \mod $ $X^n+1$ is $\bb^{\T}N^k\bc$. If $S$ is simultaneously negacyclic and totally isotropic,
 then $\ba^{\T}N^k\bd-\bb^{\T}N^k\bc=0$ for all $k$. Thus $a(X)d(X^{-1})-b(X)c(X^{-1})=0 \mxn$.
 
 Conversely, observing the constant term of the polynomial equation, we conclude that $S$ is totally isotropic.
  
\end{proof}

Let $S$ be a simultaneously negacyclic subspace of $\fpns$. Let $F$ denote the projection of $S$ onto
the first $n$ components, i.e., $F=\{\ba:(\ba,\bb)\in S\}$. We infer that $F$ is a negacyclic subspace of $\fpn$.

\begin{definition}
Let $S$ be a simultaneously negacyclic subspace of $\fpns$. Let $F$ denote the projection of $S$ onto
the first $n$ components. Let $g(X)$
be the generator of the ideal corresponding to $F$ in $\rn$.
If there is a unique polynomial $f$
 in $\rn$ such that $(\boldsymbol{g},\boldsymbol{f})\in S$, then we say that $S$ is {\it uniquely negacyclic}.
This pair of polynomials is called the {\it generating pair} for $S$.
\end{definition}

\begin{proposition}\label{gpfuns}
 Let $S$ be a simultaneously negacyclic subspace of $\fpns$. Then $S$ is uniquely negacyclic \ifif for any
 element $(\ze,\bb)$ in $S$, $\bb=\ze$. Further, if $S$ is uniquely negacyclic with the generating pair $(g,f)$,
 then $S=\{(ag,af):a\in\rn\}$.
\end{proposition}
\begin{proof}
 Assume that $S$ is uniquely negacyclic. Let $(g,f)$ be the generating pair for $S$. If $(\ze,\bb)\in S$
 for some $\bb\neq\ze$, then $(\boldsymbol{g},\boldsymbol{f}+\bb)$ belongs to $S$ and which is a contradiction.
 
 Conversely, assume that for any $(\ze,\bb)\in S$, $\bb=\ze$ holds. Let $g(X)$ be the generator of
 the corresponding ideal of
 $A=\{\ba:(\ba,\bb)\in S\}$. Suppose, $f_1(X)$ and $f_2(X)$ in $\rn$ be such that $(\boldsymbol{g},\boldsymbol{f_1})$
 and $(\boldsymbol{g},\boldsymbol{f_2})$ both are in $S$. Then $(\ze,\boldsymbol{f_1}-\boldsymbol{f_2})$ is in $S$.
 Hence, $\boldsymbol{f_1}-\boldsymbol{f_2}=\ze$. Therefore $f_1(X)=f_2(X)$.
 
 For the second part, observe that $(X^kg,X^kf)$ is in $S\subset\rn\times\rn$ for all $k$,
 because $S$ is simultaneously
 negacyclic. Thus $(ag,af)\in S$ for any polynomial $a(X)$ in $\rn$.
 Let $(c,d)$ be any element in $S$. There exists some $a(X)$ in $\rn$ such that $c=ag$ and $(ag,af)\in S$. Thus
 $(0,d-af)$ belongs to $S$. Therefore, $d=af$. This completes the proof.
  
\end{proof}

The totally isotropic subspace $S$ associated to a CSS code is of the form $C_1\times C_2$ where $C_1$ and $C_2$ are
two classical codes over $\fp$ of same length. So the elements of the form $(\ba,\ze)$, $\ba\in C_1$ and $(\ze,\bb)$,
$\bb\in C_2$ are all in $S$. These observations lead us to the following proposition:
\begin{proposition}
 Suppose $S$ is the totally isotropic subspace of a stabilizer code $\qc(\s)$. If $S$ is uniquely negacyclic, then
 $\qc(\s)$ is not CSS unless it is of distance $1$.
\end{proposition}

For a uniquely negacyclic subspace the totally isotropic condition simplifies to the following form:
\begin{proposition}\label{icip}
 Let $S$ be a uniquely negacyclic subspace of $\fpns$ with the generating pair $(g,f)$. The subspace $S$
 is totally isotropic \ifif the following condition holds:
 \begin{equation}
  g(X)f(X^{-1})=f(X)g(X^{-1}) \mxn.
 \end{equation}
 Further, any element $(\ba,\bb)$ of $\fpns$ is in $S^{\perp}$ \ifif the following condition holds:
 \begin{equation}
  a(X)f(X^{-1})=b(X)g(X^{-1}) \mxn.
 \end{equation}
\end{proposition}
\begin{proof}
 Let $S$ be a uniquely negacyclic subspace of $\fpns$ with the generating pair $(g,f)$. So
 $S=\{(ag,af):a\in\rn\}$ and $S$ is totally isotropic \ifif for any two elements $(ag,af)$ and $(bg,bf)$,
 where $a,b\in\rn$, we have
 \begin{equation*}
  a(X)g(X)b(X^{-1})f(X^{-1})-a(X)f(X)b(X^{-1})g(X^{-1})=0 \mxn.
 \end{equation*}
 This is true \ifif $g(X)f(X^{-1})=f(X)g(X^{-1}) \mxn$.
 
 For the second part, $(\ba,\bb)\in S^{\perp}$ \ifif
 \begin{equation*}
  a(X)c(X^{-1})f(X^{-1})-b(X)c(X^{-1})g(X^{-1})=0 \mxn \mbox{ for all } c\in\rn,
 \end{equation*}
 and this holds \ifif $a(X)f(X^{-1})=b(X)g(X^{-1}) \mxn$.
  
\end{proof}

Let $\fps=\fpe$ be a quadratic extension of $\fp$, where $\eta$ is a root of some irreducible quadratic polynomial
over $\fp$. The map $(\ba,\bb)\longmapsto\ba+\eta\bb$ is an isomorphism from the vector space $\fpns$ over $\fp$
to the vector space $\fpsn$ over
$\fp$. The latter vector space, i.e., $\fpsn$ is also
a vector space over the field $\fps$.
The elements of the product ring $\rn\times\rn$ can be identified with the elements of
$\rn(\eta):=\fps[X]/\ideal{X^n+1}$ using the map $(\ba,\bb)\longmapsto\ba+\eta\bb$.
Let $\qc(\s)$ be the stabilizer code of the totally isotropic subspace $S$ of $\fpns$.
If $\stil$ denote the image of $S$ in $\fpsn$, then the quantum code $\qc(\s)$ is said to be {\it linear} if
$\stil$ is a subspace of $\fpsn$ over the field $\fps$. The image $\stil$ of $S$
will also be denoted by $S$ again.
It is easy to check that if any simultaneously negacyclic subspace $S$ is a subspace of $\fpsn$ over
$\fps$, then $S$ is a negacyclic subspace of $\fpsn$ over $\fps$. The following proposition gives an important
characterization of linear stabilizer codes:
\begin{proposition}\label{lei}
 Let $S$ be a totally isotropic simultaneously negacyclic subspace of $\fpns$.
 The stabilizer
 code $\qc(\s)$ is linear \ifif $S$ is an ideal of the ring $\rn(\eta)$. If
 $\qc(\s)$ is linear, then the dual $S^{\perp}$ of $S$ is also an ideal of
 $\rn(\eta)$.
\end{proposition}
\begin{proof}
 Since $S$ is given to be simultaneously negacyclic, $\qc(\s)$ is linear \ifif $S$
 is a negacyclic subspace of $\fpsn$ over $\fps$.
 Thus the first part of the proposition follows from the fact
 that any negacyclic cyclic subspace of $\fpsn$ over $\fps$ corresponds to an ideal of
 $\fps[X]/\ideal{X^n+1}$.
 
 From Proposition~\ref{ssp}, $S^{\perp}$ is a simultaneously negacyclic subspace of $\fpns$. Therefore,
 if we show that $S^{\perp}$ is a subspace of $\fpsn=\fpen$ over $\fps=\fpe$, then using argument similar to
 previous paragraph, $S^{\perp}$ is an ideal of $\rn(\eta)$.
 Thus we need to show, for any element $\ba+\eta\bb\in S^{\perp}$, the element
 $\eta(\ba+\eta\bb)\in S^{\perp}$.
 Let $\eta$ be a root of $X^2+c_1X+c_0$, where $c_0$ and $c_1$ are in $\fp$.
 Then $\eta(\ba+\eta\bb)=-c_0\bb+\eta(\ba-c_1\bb)$. Now for any
 element $(\bx,\by)\in S$ we have,
 \begin{equation*}
  \braket{(\bx,\by),(-c_0\bb,\ba-c_1\bb)}_s=\braket{(-c_1\bx+c_0\by,-\bx),(\ba,\bb)}_s=0,
 \end{equation*}
 since $-c_1\bx+c_0\by-\eta\bx=(-\eta-c_1)(\bx+\eta\by)\in S$. Thus
 $\eta(\ba+\eta\bb)\in S^{\perp}$. This completes the proof.
  
\end{proof}

\section{Linear \texorpdfstring{$t$}{t}-Frobenius Negacyclic Codes}
\label{LtFNC}
In this section, we study linear negacyclic quantum stabilizer codes whose length $n$ divides $p^t+1$ for some positive
integer $t$, where
$\frac{p^t+1}{n}$ is odd. Recall that $\rn$ denote the ring $\fp[X]/\ideal{X^n+1}$.
The following proposition illustrates a key simplification achieved by considering codes
of such length:
\begin{proposition}\label{pt+1}
 Suppose $n$ divides $p^t+1$ for some positive integer $t$ where $\frac{p^t+1}{n}$ is an odd integer.
 Then there exists an embedding of $\fp[X]/\ideal{X^n+1}$ into $\fp[X]/\langle X^{p^t+1}+1\rangle$,
 and hence the multiplicative inverse of $X$ in $\fp[X]/\ideal{X^n+1}$
 is equal to $-X^{p^t}$. Thus, for any polynomial $f(X)$ in $\fp[X]/\ideal{X^n+1}$, we have $f(X^{-1})=f^{p^t}(-X)$.
\end{proposition}
$\rn(\eta)=\fps[X]/\ideal{X^n+1}$ is a principal ideal ring, and any ideal of it is generated by factors of
$X^n+1$ over $\fps$.
\begin{definition}\label{tfnc}
 A $t$-Frobenius negacyclic code is defined as a negacyclic quantum stabilizer code whose length $n$
 divides $p^t+1$ and $\frac{p^t+1}{n}$ is an odd integer.
\end{definition}
\begin{lemma}\label{lguc}
 Let $S$ be the totally isotropic ideal associated to a linear $t$-Frobenius negacyclic code over $\fp$ of length $n$. If
 the generator $g(X)$ of the ideal in $\rn$ corresponding to $F:=\{\ba:(\ba,\bb)\in S\}$ satisfy the
 condition $g(-X)=g(X)$, then $S$ is uniquely
 negacyclic.
\end{lemma}
\begin{proof}
 Let $\fpe$ be the quadratic extension of $\fp$ such that $S$ corresponds to an ideal of $\fpe[X]/\ideal{X^n+1}$.
 Let $\eta$ be a root of a irreducible quadratic $X^2+c_1X+c_0$ over $\fp$.
 Since $g(X)\in F$, there exists $\boldsymbol{f}\in\fpn$, such that
 $(\boldsymbol{g},\boldsymbol{f})\in S$. Let $\boldsymbol{\tilde{f}}$ be such that $(\boldsymbol{g},
 \boldsymbol{\tilde{f}})$ is also in $S$. It is enough to show that $f=\tilde{f} \mxn$. Let $\boldsymbol{h}=
 \boldsymbol{f}-\boldsymbol{\tilde{f}}$. Thus $(\ze,\boldsymbol{h})\in S$. We show separately that $h=0$ mod
 $g$ and $h=0$ mod $\frac{X^n+1}{g}$, and then the result follows from the fact that $\gcd(g,\frac{X^n+1}{g})=1$.
 
 Since $S$ is totally isotropic, and $(\ze,\boldsymbol{h})$ and $(\boldsymbol{g},\boldsymbol{f})$ are elements
 of $S$, we have $g(X^{-1})h(X)=0 \mxn$. We also have from Proposition~\ref{pt+1}, that $g(X^{-1})=g^{p^t}(-X)=
 g^{p^t}(X)$. Therefore, $g^{p^t}h=0 \mxn$. Since $g$ has an inverse modulo $\frac{X^n+1}{g}$, we conclude that
 $h=0\mod\frac{X^n+1}{g}$.
 
 We have $\eta\boldsymbol{h}\in S$. Since $S$ is $\fpe$-linear, $\eta^2\boldsymbol{h}=-c_0\boldsymbol{h}-
 \eta c_1\boldsymbol{h}$ is in $S$. Further, any element of $S$ is of the form $(ag,b)$ for some polynomials $a,b\in
 \fp[X]/\ideal{X^n+1}$. Therefore, $-c_0h=ag$ and hence $h=0\mod g$. This completes the proof.
  
\end{proof}

Let $\fpe$ be a $d$ degree extension of $\fp$. Consider the Frobenius automorphism
$\sigma$ on this extension $\fpe$ which maps any element $\alpha$ in $\fpe$ to $\alpha^p\in\fpe$.
Note that $\sigma(\alpha)=
\alpha$ \ifif $\alpha\in\fp$. The automorphism can be extended naturally to $\fpe[X]$ by mapping
any polynomial $a(X)=a_0+a_1X+\cdots+a_{n-1}X^{n-1}$ in $\fpe[X]/\ideal{X^n+1}$ to $\sigma(a):=
\sigma(a_0)+\sigma(a_1)X+\cdots+\sigma(a_{n-1})X^{n-1}$. This is called the {\it Frobenius involution}.
\begin{lemma}\label{fl}
 Let $n$ divide $p^t+1$ for some positive integer $t$, where $\frac{p^t+1}{n}$ is an odd integer.
 \begin{enumerate}[{\rm (1)}]
  \item Any irreducible factor of $X^n+1$ over $\fp$ other than the possible linear factors has even degree.
  \item Let $f(X)$ be any irreducible factor of $X^n+1$ over $\fp$ whose degree is divisible by some positive
  integer $k$. The polynomial $f(X)$ splits into $k$ irreducible factors
  $f_0(X,\eta),\ldots,f_{k-1}(X,\eta)$ such that $f_i=\sigma^i(f_0)$
  over the extension field $\mathbb{F}_{p^k}=\fpe$.
 \end{enumerate}
\end{lemma}
\begin{proof}
 (1) Let $f(X)$ be any irreducible factor of $X^n+1$ over $\fp$ other than the possible linear factors. If
 $f(X)=X^2+1$, then it is of even degree.
 
 Assume that $f(X)\neq X^2+1$ and without loss of generality it can be chosen as monic.
 If degree of $f(X)$ is $d$, then the splitting field of $f(X)$ over $\fp$ is
 $\mathbb{F}_{p^d}$. Let $\beta$ be a root of $f(X)$ in $\mathbb{F}_{p^d}$. So, $\beta^n+1=0$
 and $n$ divides $p^t+1$ with odd quotient $\frac{p^t+1}{n}$. Therefore, $\beta^{p^t+1}+1=0$, i.e.,
 $\beta^{p^t}=-\beta^{-1}$. Moreover, $\beta^{p^t}=\sigma^t(\beta)$. Since $\sigma(f)=f$, $\sigma$ maps roots
 $f(X)$ to roots of itself. Therefore, $\beta^{p^t}=-\beta^{-1}$ is also a root of $f(X)$. If $\beta=
 -\beta^{-1}$, then $\beta^2+1=0$, which is not the case. Thus, roots of $f(X)$ occurs in pairs. Hence $f(X)$
 is of even degree.
 
 (2) Let $f(X)$ be any irreducible factor of $X^n+1$ over $\fp$ such that degree of $f(X)$ is $d=km$ for some
 positive integer $m$. Thus the splitting field of $f(X)$ is $\mathbb{F}_{p^d}=\mathbb{F}_{p^{km}}$ and hence
 it contains $\mathbb{F}_{p^k}$. Thus, degree of any irreducible factor of $f(X)$ over $\mathbb{F}_{p^k}=\fpe$
 is the degree of the extension $\mathbb{F}_{p^d}/\mathbb{F}_{p^k}$ which is $m$. Since $\sigma$ is a field
 automorphism, it maps any irreducible factor of $f(X)$ over $\mathbb{F}_{p^k}$
 to the other irreducible factors.
 Over $\mathbb{F}_{p^k}$ order of $\sigma$ is $k$. Let $f_0(X,\eta)$ be an irreducible factor of $f(X)$
 over $\mathbb{F}_{p^k}$. Clearly the product $f_0\sigma(f_0)\cdots\sigma^{d-1}(f_0)$ is of degree $km$. Thus
 $f=f_0\sigma(f_0)\cdots\sigma^{d-1}(f_0)$. This completes the proof.
  
\end{proof}

If $X-\alpha$ is a factor of $X^n+1$ over $\fp$, then $\alpha=-\alpha^{-1}$ and
hence it is an element of order $4$ in $\fp$. It is easy to verify, that $\fp$ has an element of order $4$ \ifif
$p=1\mod 4$. In this case $\fp$ has exactly two elements of order $4$, namely, the roots of $X^2+1$.
Hence, possible linear factors of $X^n+1$ over $\fp$ are $X-\alpha$ and $X+\alpha$ where $\alpha^2+1=0$
when $p=1\mod 4$.

For the rest of this section, we consider only the quadratic extension $\fps=\fpe$. Recall that a stabilizer
negacyclic code is linear \ifif the associated totally isotropic subspace
is an ideal of $\rn(\eta)=\fpe[X]/\ideal{X^n+1}$.
The following theorem gives a necessary condition for a simultaneously negacyclic subspace of $\fpns$ to be a
totally isotropic ideal of $\rn(\eta)$:
\begin{theorem}\label{ncflc}
 Let $\fp(\eta)$ be a quadratic extension of $\fp$.
 Let $S$ be a totally isotropic ideal, associated to a linear $t$-Frobenius negacyclic code,
 of $\rn(\eta)=\fp(\eta)[X]/\ideal{X^n+1}$,
 such that the generator $g(X)$ of the ideal in $\rn$ corresponding to $F:=\{\ba:(\ba,\bb)\in S\}$
 satisfies $g(-X)=g(X)$. Then the ideal $S$
 is generated by the product polynomial $g(X)h(X,\eta)$ where $h(X,\eta)$ is such a factor of $X^n+1$ that
 $g(X)$ and $h(X,\eta)$ are coprime, and the following holds:
 \begin{enumerate}[{\rm (1)}]
  \item $g(X)$ contains all possible linear factors of $X^n+1$ over $\fp$ as factors.
  \item $h(X,\eta)$ is a factor of $\frac{X^n+1}{g}$ over $\fp(\eta)$, such that for any irreducible factor
  $r(X,\eta)$ of $X^n+1$ over $\fpe$, exactly one of $r$ and $\sigma(r)$ divide $h$.
 \end{enumerate}
 In addition, if the factor $h(X,\eta)$ satisfies the condition $h(-X,\eta)=h(X,\eta)$, then $t$ must be even.
\end{theorem}
\begin{proof}
 Let $\eta$ be a root of the irreducible polynomial $c(X)=X^2+c_1X+c_0$ over $\fp$. So, $\eta^2=-c_1\eta-c_0$.
 On applying Lemma~\ref{lguc}, we conclude that $S$ is uniquely cyclic.
 Let $(g,f)$ be the generating pair for $S$.
 In particular, $g+\eta f$ belongs to $S$. Since $S$ is an ideal of $\rn(\eta)$, $\eta(g+\eta f)$ is
 also an element of $S$. From Proposition~\ref{gpfuns}, it follows that there exists $a(X)\in\fp[X]$
 such that
 \begin{equation}
  \eta(g+\eta f)=a(g+\eta f) \mxn
 \end{equation}
 Because $\eta^2=-c_1\eta-c_0$, we get,
 \begin{equation}
  -c_0f+\eta(g-c_1f)=ag+\eta af \mxn.
 \end{equation}
 Comparing the coefficients of $\eta$, we obtain
 \begin{align}
  & f=-\frac{a}{c_0}g \mxn \label{fee}\\
  \mbox{and}\quad & af+c_1f-g=0 \mxn\label{see}.
 \end{align}
 Substituting the value of $f$ from the equation \eqref{fee} in equation \eqref{see} and using the fact
 that $g$ is invertible modulo $\frac{X^n+1}{g}$, we obtain
 \begin{equation}\label{eta}
  c(a(X))=0\mod\frac{X^n+1}{g}.
 \end{equation}
 Let $r(X)$ be any irreducible factor of $\frac{X^n+1}{g(X)}$ over $\fp$. Let $\mathbb{K}$ be the splitting
 field $\fp[X]/\ideal{r(X)}$ of $r(X)$ over $\fp$. From equation \eqref{eta}, it follows that $a\mod r$
 is a root of $c(X)$ in $\mathbb{K}$ and hence $c(X)$ splits over $\mathbb{K}$. Thus $\mathbb{K}$ contains
 the quadratic extension $\fpe$. Therefore, the degree of the extension $\mathbb{K}/\fp$ is even and
 thus equal
 to the degree of the polynomial $r(X)$. This shows that any irreducible factor of $\frac{X^n+1}{g}$ over
 $\fp$ must be of even degree and as a result $g$ should contain all the odd degree factors of $X^n+1$ over $\fp$
 as factors.
 From Lemma~\ref{fl} it follows that, the possible linear factors are the only odd degree factors of
 $X^n+1$ over $\fp$. So $g$ should contain all possible linear factors of $X^n+1$ over $\fp$ as factors.
 
 Let $h(X,\eta)$ be the polynomial $\gcd\left(\frac{X^n+1}{g},1-\frac{\eta}{c_0}a\right)$ over $\fpe$.
 This $h(X,\eta)$ is a factor
 of $X^n+1$ over $\fpe$ and is coprime to $g(X)$. From Proposition~\ref{gpfuns}, we infer that, $g+\eta f$
 is a generator of $S$ as an ideal of $\rn(\eta)$. Therefore, $\gcd(X^n+1,g+\eta f)=\gcd(X^n+1,g-
 \frac{\eta}{c_0}ag)=g\gcd\left(\frac{X^n+1}{g},1-\frac{\eta}{c_0}a\right)=gh$ is also a generator of $S$. We next
 show that $h(X,\eta)$ satisfies statement (2) of the theorem.
 If $r(X)$ is a factor of $\frac{X^n+1}{g(X)}$ over $\fp$,
 then from the first part of the proof we know that $r(X)$ is of even degree.
 In Lemma~\ref{fl} we observed that such polynomial $r(X)$ factorizes
 as $r=\tilde{r}\sigma(\tilde{r})$, where $\tilde{r}(X,\eta)$ is an irreducible polynomial over $\fpe$.
 But $a\mod\tilde{r}$ in $\fpe[X]/\ideal{\tilde{r}(X,\eta)}$ is a root of $c(Y)$. Therefore,
 $a\mod\tilde{r}$ is either $\eta$ or $\eta'$, where $\eta'=\sigma(\eta)$. It is immediate that,
 $\tilde{r}$ divides $h$ \ifif $1-\frac{\eta}{c_0}a=0\mod\tilde{r}$. This holds \ifif
 $a=\eta'\mod\tilde{r}$, since $\eta\eta'=c_0$. Similarly, $\sigma(\tilde{r})$ divides $h$ \ifif
 $a=\eta'\mod\sigma(\tilde{r})$. Now, $a=\eta'\mod\sigma(\tilde{r})$ \ifif $\sigma(a)=
 \sigma(\eta')\mod\sigma^2(\tilde{r})$. Also we have, $\sigma(a)=a$ and $\sigma^2$ is the identity
 on $\fpe$. Thus, $\sigma(\tilde{r})$ divides $h$ \ifif $a=\eta\mod\tilde{r}$. Therefore, exactly
 one of $\tilde{r}$ and $\sigma(\tilde{r})$ divides $h$ depending on whether $a=\eta'$ or $\eta$ modulo
 $\tilde{r}$ respectively. This proves the statement (2) of the theorem.
 
 Next, assume that $h$ satisfies the property $h(-X,\eta)=h(X,\eta)$. From the second part of the proof,
 we have $h(X,\eta)=\gcd\left(\frac{X^n+1}{g},1-\frac{\eta}{c_0}a\right)$. Therefore, in particular $h$ divides
 $1-\frac{\eta}{c_0}a$. Thus, $a=\eta'\mod h$ and by applying $\sigma$ we obtain $a=\eta\mod \sigma(h)$.
 Now, $a(-X)=\eta'\mod h(-X,\eta)=\eta'\mod h(X,\eta)$ and similarly, $a(-X)=\eta\mod \sigma(h(X,\eta))$.
 Thus, by combining we have $a(-X)=a(X)\mod\frac{X^n+1}{g}$, because $h$ and $\sigma(h)$ are relatively prime
 and $\frac{X^n+1}{g}=h\sigma(h)$. Since $S$ is totally isotropic and $(g,f)\in S$,
 $g(X)f(X^{-1})=f(X)g(X^{-1}) \mxn$. Using equation \eqref{fee} and Proposition~\ref{pt+1} we obtain, 
 $g^{p^t+1}(X)a^{p^t}(-X)$ $=g^{p^t+1}(X)a(X) \mxn$, since $g(-X)=g(X)$. Because $g$ is invertible
 modulo $\frac{X^n+1}{g}$, $a^{p^t}(-X)=a(X)\mod\frac{X^n+1}{g}$. Therefore,
 \begin{equation}\label{apta}
  a^{p^t}(X)=a(X)\mod\frac{X^n+1}{g}
 \end{equation}
 As in first part of the proof, let $r(X)$ be any irreducible factor
 of $\frac{X^n+1}{g}$ and $\mathbb{K}$ be the field $\fp[X]/\ideal{r(X)}$. Then from equation \eqref{eta}
 $a\mod r$ is a root of $c(Y)$ in $\mathbb{K}$. If possible suppose, $t=2m+1$ for some positive
 integer $m$. We remark that any extension field $\mathbb{F}_{p^{2m}}/\fp$ contains $\fpe$. Thus $c(Y)$
 divides $Y^{p^{2m}}-Y$. Therefore, $a^{p^{2m}}=a\mod r$. Thus from equation \eqref{apta}, we derive that
 $a^p=a\mod r$. Therefore, $a\mod r$ is an element of $\fp\subset\mathbb{K}$. This is a
 contradiction, because $c$ is irreducible over $\fp$. Hence in this case, $t$ must be even.
 This completes the proof.
  
\end{proof}

\section{Construction of Linear and Nonlinear Codes}
\label{CLNlC}
Finally in this section we establish certain sufficient conditions to obtain $t$-Frobenius negacyclic codes and
study the BCH distance for such codes.
\begin{theorem}\label{gnqc}
 Let $n$ divide $p^{km}+1$ for some positive integers $k$ and $m$, where $\frac{p^{km}+1}{n}$ is an odd
 integer. Let $\fpe$ be a degree $k$ extension of $\fp$. Let $g(X)$ and $h(X,\eta)$ be coprime factors of
 $X^n+1$ satisfying the following properties:
 \begin{enumerate}[{\rm (1)}]
  \item $g(X)$ is any factor of $X^n+1$ over $\fp$ such that $g(-X)=g(X)$ and $g(X)$ contains all the irreducible
  factors of $X^n+1$ over $\fp$ whose degree is not divisible by $k$.
  \item $h(X,\eta)$ is any factor of $\frac{X^n+1}{g}$ over $\fpe$ such that $h(-X,\eta)=h(X,\eta)$ and for any
  irreducible factor $r(X,\eta)$ of $\frac{X^n+1}{g}$ over $\fpe$, $r(X,\eta)$ divides $h(X,\eta)$ \ifif
  the factors $\sigma^i(r)$ does not divide $h$ for $i=1,2,\ldots,k-1$, i.e., $\frac{X^n+1}{g}=
  \prod_{i=0}^{k-1}\sigma^i(h)$.
 \end{enumerate}
 Choose any nonzero $\alpha$ in $\fp$ and let $a(X,\eta)$ be the polynomial, uniquely defined by Chinese
 remaindering, as follows:
 \begin{equation*}
  a=\begin{cases}
     1 & \mod g,\\
     \sigma^i(\alpha\eta) & \mod\sigma^i(h) \mbox{ for all $0\leqslant i<k$}.
    \end{cases}
 \end{equation*}
 Then $a(X,\eta)\in\fp[X]$ and the uniquely negacyclic subspace generated by $(g,ag)$ is
 totally isotropic.
\end{theorem}
\begin{proof}
 Let $f(X)$ denote the polynomial $\frac{X^n+1}{g(X)}$. Clearly $f=\prod_{i=0}^{k-1}\sigma^i(h)$ over $\fpe$.
 Assume $\alpha$ and $a(X,\eta)$ is chosen as stated in the theorem. The following three claims lead us to
 the proof of the theorem:
 \begin{claim}
  The polynomial $a(X,\eta)$ is a polynomial over $\fp$ and $a(-X,\eta)=a(X,\eta)$.
 \end{claim}
 \proof It is enough to show that $\sigma(a)=a$. Since $\sigma(g)=g$ and $a=1\mod g$, we obtain
 $\sigma(a)=a\mod g$. From the assumption of the theorem, we have $a=\sigma^i(\alpha\eta)
 \mod\sigma^i(h)$ for all $0\leqslant i<k$. Therefore, $\sigma(a)=\sigma^{i+1}(\alpha\eta)
 \mod\sigma^{i+1}(h)$ for all $0\leqslant i<k$. Since $\sigma^k$ is the identity on $\fpe$, $\sigma(a)=
 \sigma^i(\alpha\eta)\mod\sigma^i(h)$ for all $0\leqslant i<k$. Thus $\sigma(a)=a\mod\sigma^i(h)$
 for all $0\leqslant i<k$. The first part of the claim follows, since $g$ and $\sigma^i(h)$ for all
 $0\leqslant i<k$ are pairwise coprime.
 
 For the second part also, we verify the statement modulo $g$ and $\sigma^i(h)$ for all $0\leqslant i<k$
 separately. Since $g(-X)=g(X)$ and $a(X)=1\mod g(X)$, we obtain
 $a(-X)=a(X)\mod g(X)$. Since $h(-X,\eta)=h(X,\eta)$, it follow that
 $\sigma^i(h(-X,\eta))=\sigma^i(h(X,\eta))$.
 Therefore, $a(-X,\eta)=\sigma^i(\alpha\eta)\mod\sigma^i(h(X,\eta))$ for all $0\leqslant i<k$, i.e.,
 $a(-X,\eta)=a(X,\eta)\mod\sigma^i(h(X,\eta))$ for all $0\leqslant i<k$. Hence the second part follows.
 \begin{claim}\label{axi=ax}
  Since $a(X,\eta)$ is a polynomial over $\fp$, we write it as $a(X)$. Then $a(X^{-1})=a(X)\mxn$.
 \end{claim}
 \proof From Proposition~\ref{pt+1}, we have $a(X^{-1})=a^{p^{km}}(-X)\mxn$. Using second part of the
 previous claim, we infer that $a(X^{-1})=a^{p^{km}}(X)\mxn$. From the definition of $a$, it follows
 that $a^{p^{km}}=
 (\sigma^i(\alpha\eta))^{p^{km}}\mod\sigma^i(h)$ for all $0\leqslant i<k$, i.e., $a^{p^{km}}=
 \sigma^i(\alpha\eta) \mod\sigma^i(h)$ for all $0\leqslant i<k$. Similarly, $a^{p^{km}}=1\mod g$.
 Hence, $a(X^{-1})=a(X)\mxn$.
 \begin{claim}
  The uniquely negacyclic subspace generated by $(g,ag)$ is totally isotropic.
 \end{claim}
 \proof From Proposition~\ref{icip}, it is sufficient to prove
 \begin{equation}
  g(X)a(X^{-1})g(X^{-1})-a(X)g(X)g(X^{-1})=0 \mxn,
 \end{equation}
 i.e., it is sufficient to prove
 \begin{equation}
  g^{p^{km}+1}(X)a(X^{-1})=g^{p^{km}+1}(X)a(X)\mxn.
 \end{equation}
 This holds because of the previous claim, i.e., because $a(X^{-1})=a(X)\mxn$.
 
 This completes the proof of the theorem.
  
\end{proof}

If $S$ is the uniquely negacyclic totally isotropic subspace generated by $(g,ag)$ as in Theorem~\ref{gnqc},
then the corresponding stabilizer code $\qc(\s)$ is a $km$-Frobenius negacyclic code.
We refer to the product $g(X)\cdot h(X,\eta)$ of polynomials obtained in Theorem~\ref{gnqc} for any
$km$-Frobenius negacyclic code as the {\it canonical factorization} associated to the code.
The next theorem is the converse of Theorem~\ref{ncflc}.

\begin{theorem}\label{lcfgc}
 Suppose $c(X)=X^2+c_1X+c_0$ is an irreducible polynomial over $\fp$ and $\eta,\eta'$ be roots of $c(X)$. Set
 $k=2$, the quadratic extension $\fp(\eta)=\fp(\eta')$ and $\alpha
 =-c_0^{-1}$ in Theorem~\ref{gnqc}. Let $S$ be the corresponding totally isotropic subspace. Then the
 image of $S$ under the map $(u,v)\mapsto u+\eta v$ is an ideal of $\rn(\eta)$ and is generated by
 $g(X)h(X,\eta)$, where $g$ and $h$ satisfies the properties of Theorem~\ref{ncflc}. Further, the associated
 stabilizer code $\qc(\s)$ is a linear negacyclic code and
 the dual $S^{\perp}$ of $S$ \wrt the symplectic inner product also maps to an
 ideal generated by $h(X,\eta)$.
\end{theorem}
\begin{proof}
 From Theorem~\ref{gnqc}, it follows that the polynomial $a(X,\eta)$ defined by
 \begin{equation}\label{lca}
  a=\begin{cases}
     1 & \mod g,\\
     -c_0^{-1}\eta' & \mod h,\\
     -c_0^{-1}\eta & \mod\sigma(h).
    \end{cases}
 \end{equation}
 is a polynomial over $\fp$ and the uniquely negacyclic subspace $S$ generated by $(g,ag)$ is totally
 isotropic. The elements of $S$ are precisely of the form $(ug,uag)$
 where $u\in\rn$. To show $S$ is an ideal of $\rn(\eta)$, it is enough to show $\eta(ug+\eta uag)\in S$,
 since $S$ is uniquely negacyclic. This follows from the following claim.
 \begin{claim}
  $\eta(g+\eta ag)=-c_0a(g+\eta ag)\mxn$.
 \end{claim}
 \proof We verify the statement modulo $g,h$ and $\sigma(h)$ separately.
 Since $g$ is a factor on both the side of the equation in the claim, the equation holds modulo
 $g$. We have $\eta^2=-c_1\eta-c_0$ and $\eta\eta'=c_0$.
 Using the the definition of $a$ as stated in \eqref{lca}, we have
 \begin{align*}
  \eta(g+\eta ag)+c_0a(g+\eta ag)
  &= \eta(g+\eta(-c_0^{-1}\eta')g)+c_0(-c_0^{-1}\eta')(g+\eta(-c_0^{-1}\eta')g) \mod h\\
  &= \eta(g-g)+\eta'(g-g) \mod h=0 \mod h.
 \end{align*}
 Similarly,
 \begin{align*}
  \eta(g+\eta ag)+c_0a(g+\eta ag)
  &= \eta(g+\eta(-c_0^{-1}\eta)g)+c_0(-c_0^{-1}\eta)(g+\eta(-c_0^{-1}\eta)g) \mod h\\
  &= \eta(g-\eta^2c_0^{-1}g)-\eta(g-\eta^2c_0^{-1}g) \mod h=0 \mod h.
 \end{align*}
 This proves the claim.
 
 Thus $\eta(ug+\eta uag)=-c_0au(g+\eta ag)\mxn$ belongs to $S$. Hence it
 follows that the image of $S$ under the map $(u,v)\mapsto u+\eta v$ is an ideal of $\rn(\eta)$.
 We now show that $S$ is generated by $g(X)h(X,\eta)$. Since $g+\eta ag$ is a generator of $S$,
 $\gcd(X^n+1,g+\eta ag)=g\gcd\left(\frac{X^n+1}{g},1+\eta a\right)$ is also a generator.
 We show that $h(X,\eta)=\gcd\left(\frac{X^n+1}{g},1+\eta a\right)$.
 For this first observe that for any irreducible factor $r(X,\eta)$ of
 $\frac{X^n+1}{g}$ over $\fpe$,
 \begin{equation}\label{mlca}
  a\mod r=\begin{cases}
           -\eta^{-1} & \mbox{ if } r|h,\\
           -c_0^{-1}\eta & \mbox{ if } r|\sigma(h).
          \end{cases}
 \end{equation}
 Therefore, if $r$ divide $h$, then $1+\eta a=0\mod r$. Also if $r$ does not divide $h$, then it divide
 $\sigma(h)$ and in this case $1+\eta a=(2+\frac{c_1}{c_0}\eta)\mod r\neq0$. Combining this two, we get,
 $r$ divide $h$ \ifif $r$ divide $1+\eta a$. Thus, $\gcd\left(\frac{X^n+1}{g},1+\eta a\right)=h(X,\eta)$.
 
 Let $I$ denote the ideal generated by $1+\eta a$ in $\rn(\eta)$. Since $1+\eta a=(1+\eta)
 \mod r\neq0\mod r$ for any irreducible factor $r$ of $g$, the polynomials $1+\eta a$ and $g$
 are relatively prime. Thus, from the
 same arguments as above it follows that $\gcd(X^n+1,1+\eta a)=h$. Hence $I$ is also generated by $h$.
 We have
 \begin{equation*}
  g(X)a(X^{-1})-a(X)g(X)=0\mxn,
 \end{equation*}
 since $a(X^{-1})=a(X)\mxn$. Thus, the symplectic inner product of $(\boldsymbol{1},\ba)$ and
 $(\boldsymbol{g},\ba\boldsymbol{g})$ is $0$ and hence
 $1+\eta a$ belongs to $S^{\perp}$. It is shown in Proposition~\ref{lei} that $S^{\perp}$ is an
 ideal of $\rn(\eta)$. Therefore, $I\subset S^{\perp}$. We now show $I=S^{\perp}$ by verifying that
 they have same cardinality. Since $S$ is generated by $gh$ as an ideal of $\rn(\eta)$, $\{gh,Xgh,\ldots,
 X^{n-\deg(gh)-1}gh\}$ gives a basis of $S$ as a subspace of $\rn(\eta)$ over $\fpe$. Thus, the cardinality of
 $S$ is $p^{2(n-\deg(gh))}=p^{n-\deg(g)}$, since $\deg(g)+2\deg(h)=n$. Therefore, the dimension of $S$ as a
 subspace of $\fpns$ over $\fp$ is $n-\deg(g)$. Thus, the dimension of $S^{\perp}$ as a subspace of
 $\fpns$ over $\fp$ is $n+\deg(g)$ and hence the cardinality of $S^{\perp}$ is $p^{n+\deg(g)}$.
 Since $I$ is an ideal of $\rn(\eta)$ generated by $h$, the cardinality of
 $I$ is $p^{2(n-\deg(h))}=p^{n+\deg(g)}$. Therefore, $I=S^{\perp}$ and hence it is generated by $h$.
  
\end{proof}

We refer the factorization $g(X)\cdot h(X,\eta)$ in the above theorem as the {\it canonical factorization}
associated to the linear code.
Next we define the BCH distance. Since $X^n+1=\frac{X^{2n}-1}{X^n-1}$, the roots
of $X^n+1$ are the roots of $X^{2n}-1$ which are not roots of $X^n-1$ in some extension of $\fp$. Let $\beta$ be
a primitive $2n$-th root of unity in some extension of $\fp$. So $\beta$ is a root of $X^n+1$, and
$\alpha=\beta^2$ is a primitive $n$-th root of unity. Hence the roots of $X^n+1$ are $\beta\alpha^i=\beta^{1+2i}$,
$0\leqslant i\leqslant n-1$.
\begin{definition}
 Let $q=p^k$ where $p>2$ is a prime number and let $n$ be such that $\gcd(n,q)=1$.
 Let $f(X)$ be a factor of $X^n+1$ over the
 field $\mathbb{F}_q$. The BCH distance of $f(X)$ is defined to be the largest integer $d$ such that
 $\beta^{\ell},\beta^{\ell+2},\beta^{\ell+4},\ldots,\beta^{\ell+2(d-2)}$ are roots of $f(X)$ for some
 $\ell\in\{1,3,5,\ldots,2n-1\}$ and for some primitive $2n$-th root of unity $\beta$.
\end{definition}

We restate below the Lemma~4 of \cite{ks-pcc} in terms of our new definition of BCH distance. It
gives a lower bound of the minimum Hamming distance of a classical
negacyclic code in terms of the generator's BCH distance.
\begin{theorem}\label{ks}
 Let $f(X)$ be a factor of $X^n+1$ over the field $\mathbb{F}_q$. Let $C$ be the $q$-ary classical negacyclic code
 of length $n$ generated by $f(X)$. If $f(X)$ has BCH distance $d$, then the minimum distance of the code $C$ is
 at least $d$.
\end{theorem}

The {\it minimum distance} of a quantum stabilizer code $\qc(\s)$ is the minimum joint weight of
$S^{\perp}\ssm S$ for any totally isotropic subspace $S$.
\begin{definition}
 Let $\qc(\s)$ be a linear negacyclic quantum stabilizer code for any totally isotropic subspace $S\subset\fpns$. The BCH
 distance of $\qc(\s)$ is defined to be the BCH distance of the generator polynomial of $S^{\perp}$.
\end{definition}

\begin{theorem}
 Let the BCH distance of a linear negacyclic quantum stabilizer code $\qc(\s)$ be $d$.
 Then the minimum distance of $\qc(\s)$ is at least $d$.
\end{theorem}
\begin{proof}
 Let $S$ be the totally isotropic set of $\qc(\s)$. From Theorem~\ref{ks}, it follows that the
 minimum Hamming weight of $S^{\perp}$ is at least $d$. Hence, minimum joint weight of $S^{\perp}\ssm S$
 is at least $d$. Thus, the minimum distance of $\qc(\s)$ is at least $d$.
  
\end{proof}

For a $q$-ary quantum stabilizer code $\qc$, the dimension of $\qc$ is of the form $q^k$ for some nonnegative
integer $k$. The integer $k$ is referred as the stabilizer dimension of the code $\qc$.
In the case of nonlinear codes, we also define
the BCH distance of the code as the BCH distance of the polynomial $h(X,\eta)$.
\begin{theorem}\label{mdfgc}
 Let $g(X)h(X,\eta)$ be the canonical factorization associated with a $km$-Frobenius negacyclic code $\qc$ as in
 Theorem~\ref{gnqc}. Then the stabilizer dimension of $\qc$ is $\deg(g)$.
\end{theorem}
\begin{proof}
 Let $a(X)$ be the polynomial defined as in Theorem~\ref{gnqc}. Let $S$ be the totally isotropic
 subspace associated to $\qc$ which is a uniquely negacyclic subspace generated by $(g,ag)$ as
 stated in Theorem~\ref{gnqc}.
 From Proposition~\ref{gpfuns}, we have $S=\{(ug,uag)\in\rn\times\rn : u\in\rn\}$.
 It is easy to check that $(ug,uag)
 =(vg,vag)\mxn$ for some $u,v\in\rn$ \ifif $ug=vg\mxn$. Thus, the cardinality of $S$ is same as the cardinality of the
 ideal $\ideal{g}\subset\rn$. Now, $\{g,Xg,\ldots,X^{n-\deg(g)-1}g\}$ is a basis of $\ideal{g}$ as a subspace of
 $\rn$ over $\fp$. Thus, the cardinality of $S$ is $p^{n-\deg(g)}$ and hence the dimension of $S$ as a subspace of
 $\fpns$ over $\fp$ is $n-\deg(g)$. Therefore, the stabilizer dimension of $\qc$ is $\deg(g)$.
  
\end{proof}

\begin{corollary}\label{mdlnc}
 Let $\qc$ be a linear $2m$-Frobenius negacyclic code over $\fp$ with canonical factorization $g\cdot h$. Then the
 stabilizer dimension of $\qc$ is $\deg(g)$. Also, the dual $S^{\perp}$ of $S$ \wrt the symplectic inner product
 is the ideal generated by $h(X,\eta)$, and hence the minimum distance of $\qc$ is at least the BCH distance
 of $h$.
\end{corollary}

We now show that the same result holds even for nonlinear negacyclic codes.
\begin{theorem}
 Let $g(X)h(X,\eta)$ be the canonical factorization associated with a $km$-Frobenius negacyclic code $\qc$ as in
 Theorem~\ref{gnqc}. If the BCH distance of $h(X,\eta)$ is $d$, then the minimum distance of $\qc$
 is at least $d$.
\end{theorem}
\begin{proof}
 To prove this theorem, we first prove the following claim:
 \begin{claim}
  The dual $S^{\perp}$ of $S$ \wrt the symplectic inner product satisfies
  \begin{equation*}
   S^{\perp}=\left\{\left(u,au+v\frac{X^n+1}{g}\right)\in\rn\times\rn : u,v\in\rn\right\}.
  \end{equation*}
 \end{claim}
 \proof Assume $\tilde{S}$ to be the set
 $\left\{\left(u,au+v\frac{X^n+1}{g}\right)\in\rn\times\rn : u,v\in\rn\right\}$. First we show that
 $\tilde{S}\subset S^{\perp}$. For $u,v\in\rn$, using Proposition~\ref{pt+1} and Claim~\ref{axi=ax}
 we obtain
 \begin{align*}
  & u(X)a(X^{-1})g(X^{-1})-\left(a(X)u(X)+v(X)\frac{X^n+1}{g(X)}\right)g(X^{-1})\\
  =& u(X)a(X^{-1})g^{p^{km}}(X)-\left(a(X)u(X)+v(X)\frac{X^n+1}{g(X)}\right)g^{p^{km}}(X)\mxn\\
  =& u(X)a(X^{-1})g^{p^{km}}(X)-u(X)a(X)g^{p^{km}}(X)\mxn\\
  =& 0\mxn.
 \end{align*}
 Because $S$ is generated by $(g,ag)$, from Proposition~\ref{icip} and the above computations we conclude that
 $\left(u,au+v\frac{X^n+1}{g}\right)\in S^{\perp}$ for all $u,v\in\rn$. From the first part of the proof it
 follows that the dimension of $S^{\perp}$ is $n+\deg(g)$ as a subspace of $\fpns$ over $\fp$. Thus, the
 cardinality of $S^{\perp}$ is $p^{n+\deg(g)}$. To find the cardinality of $\tilde{S}$, observe that
 $au=av\mxn$ for some $u,v\in\rn$ \ifif $u=v$. Let $B$ denote the set $\left\{v\frac{X^n+1}{g}\in\rn : v\in\rn\right\}$.
 Thus the cardinality of $\tilde{S}$ is the product of the cardinalities of $\rn$ and $B$. Note that $B$ is an ideal
 of $\rn$ generated by $\frac{X^n+1}{g}$. By arguing using a basis of $B$, we find that
 the cardinality of $B$ is $p^{n-(n-\deg(g))}
 =p^{\deg(g)}$. Therefore, the cardinality of $\tilde{S}=p^np^{\deg(g)}=p^{n+\deg(g)}$. Hence $\tilde{S}=S^{\perp}$.
 This proves the claim.
 
 To prove that the minimum distance of $\qc$ is at least $d$, it is enough
 to show that the minimum joint weight of $S^{\perp}$
 is at least $d$. Let $(\bu,\bv)\in S^{\perp}$. The joint weight of
 $(\bu,\bv)$ is same as the Hamming weight of $\alpha\eta\bu-\bv\in\fpen$, and $\alpha$ be as in
 Theorem~\ref{gnqc}. Thus it only remains to show that $h(X,\eta)$ divides $\alpha\eta u(X)-v(X)$, because from
 Theorem~\ref{ks} it follows that the Hamming weight of $\alpha\eta\bu-\bv$ is at least $d$.
 From the previous claim it follows that there exists $\tilde{v}\in\rn$ such that $v=au+\tilde{v}\frac{X^n+1}{g}$.
 On substituting this value of $v$ and $a=\alpha\eta\mod h$, we obtain
 \begin{align*}
  \alpha\eta u-v &= \alpha\eta u-\left(au+\tilde{v}\frac{X^n+1}{g}\right)\\
  &=-\tilde{v}\frac{X^n+1}{g} \mod h=0 \mod h.
 \end{align*}
 This completes the proof of the theorem.
  
\end{proof}

To illustrate this construction, here we discuss an example of linear negacyclic quantum stabilizer codes over the
field $\mathbb{F}_3$. We fix $n=3^2+1$ and the quadratic extension $\mathbb{F}_3(\eta)$, where
$\eta$ satisfy the irreducible $X^2+1$. It can be seen that $X^{10}+1$ splits as $(X^2+1)(X^4+X^3+2X+1)
(X^4+2X^3+X+1)$ over $\mathbb{F}_3$. Since there is no linear factor, we choose $g(X)=X^2+1$ for this example
which satisfy the condition $g(-X)=g(X)$. Further, the factorization of $\frac{X^{10}+1}{X^2+1}$ over
$\mathbb{F}_3(\eta)$ is $(X^2+(\eta+2)X+2)(X^2+(2\eta+2)X+2)(X^2+(\eta+1)X+2)(X^2+(2\eta+1)X+2)$. We
choose $h(X,\eta)=(X^2+(\eta+2)X+2)(X^2+(2\eta+1)X+2)=X^4+(2\eta+1)X^2+1$. It follows that $h(-X,\eta)=h(X,\eta)$ and
$gh\sigma(h)=X^{10}+1$. Therefore, $gh$ gives a canonical factorization for a linear negacyclic quantum stabilizer code
as in Theorem~\ref{lcfgc}. To obtain the BCH distance of $h$, fix a root $\beta$ of $(X^2+(\eta+2)X+2)$. Then
$\beta$ is a primitive $20$-th root of unity, and it can be checked that $\beta^9$ and $\beta^{11}$ are two
roots of $h$. Thus BCH distance of $h$ is $3$. Note that the quantum code obtain here is a $10$ qubit code and
its stabilizer dimension is $\deg(g)=2$. Hence we get a $[[10,2,3]]_3$ code.

In the following three tables, we list down some more examples obtained in similar manner. These tables contain
both linear and nonlinear codes over $\mathbb{F}_3$, $\mathbb{F}_5$ and $\mathbb{F}_7$ respectively.
The nonlinear codes are distinguished by the superscript symbol ``star''.
Note that the distances given in these tables are BCH distances.

\begin{table}[h]
 \centering
 \begin{minipage}[b]{0.33\textwidth}
  \centering
  \begin{tabular}{|m{1.2cm}|m{2cm}|}
   \hline
   Length & Parameters\\ \hline
   $10$ & $[[10,2,3]]_3$\\ \hline
   $28$ & $[[28,4,3]]_3^*$, $[[28,16,3]]_3^*$\\ \hline
   $34$ & $[[34,2,4]]_3$\\ \hline
   $50$ & $[[50,2,4]]_3$, $[[50,10,3]]_3$\\ \hline
   $58$ & $[[58,2,5]]_3$\\ \hline
   $76$ & $[[76,4,3]]_3^*$, $[[76,40,3]]_3^*$\\ \hline
   $82$ & $[[82,2,7]]_3$, $[[82,18,6]]_3$, $[[82,34,4]]_3$, $[[82,50,3]]_3$, $[[82,66,3]]_3$\\ \hline
  \end{tabular}
  \caption{Codes over $\mathbb{F}_3$}
  \label{cof3}
 \end{minipage}
 \begin{minipage}[b]{0.33\textwidth}
  \centering
  \begin{tabular}{|m{1.2cm}|m{2cm}|}
   \hline
   Length & Parameters\\ \hline
   $14$ & $[[14,2,3]]_5^*$\\ \hline
   $18$ & $[[18,6,3]]_5^*$\\ \hline
   $26$ & $[[26,2,5]]_5$, $[[26,10,3]]_5$, $[[26,18,3]]_5$\\ \hline
   $34$ & $[[34,2,4]]_5$\\ \hline
   $42$ & $[[42,6,5]]_5^*$, $[[42,18,3]]_3^*$, $[[42,30,3]]_5^*$\\ \hline
   $54$ & $[[54,6,5]]_5^*$, $[[54,18,3]]_5^*$\\ \hline
   $74$ & $[[74,2,5]]_5$\\ \hline
   $82$ & $[[82,2,7]]_5$, $[[82,42,4]]_5$\\ \hline
  \end{tabular}
  \caption{Codes over $\mathbb{F}_5$}
  \label{cof5}
 \end{minipage}
 \begin{minipage}[b]{0.33\textwidth}
  \centering
  \begin{tabular}{|m{1.2cm}|m{2cm}|}
   \hline
   Length & Parameters\\ \hline
   $10$ & $[[10,2,3]]_7$\\ \hline
   $26$ & $[[26,2,3]]_7^*$, $[[26,2,5]]_7$\\ \hline
   $34$ & $[[34,2,4]]_7$\\ \hline
   $50$ & $[[50,2,7]]_7$, $[[50,10,6]]_7$, $[[50,18,5]]_7$, $[[50,26,4]]_7$, $[[50,34,3]]_7$, $[[50,42,3]]_7$\\ \hline
   $82$ & $[[82,2,6]]_7$\\ \hline
  \end{tabular}
  \caption{Codes over $\mathbb{F}_7$}
  \label{cof7}
 \end{minipage}
\end{table}

\subsection*{Acknowledgements}
The second author was supported by the Seed Grant from IRCC, IIT Bombay.


\bibliography{tFNCBib}
\bibliographystyle{plain}

\mbox{}\\ \\
Priyabrata Bag\\
Department of Mathematics\\
Indian Institute of Technology Bombay\\
Mumbai, Maharashtra 400076, India\\
E-mail: {\it priyabrata@iitb.ac.in}\\ \\
Santanu Dey\\
Department of Mathematics\\
Indian Institute of Technology Bombay\\
Mumbai, Maharashtra 400076, India\\
E-mail: {\it santanudey@iitb.ac.in}

\end{document}